\documentclass[letterpaper, 10 pt, conference]{ieeeconf}  

\IEEEoverridecommandlockouts                              

\usepackage{cite}
\usepackage{amsmath,amssymb,amsfonts}
\usepackage{algorithmic}
\usepackage{graphicx}
\usepackage{algorithm,algorithmic}
\usepackage{hyperref}
\hypersetup{hidelinks=true}
\usepackage{textcomp}
\usepackage{mathrsfs}
\usepackage{mathtools}
\usepackage{graphicx}
\usepackage{textcomp}
\usepackage{algorithm}
\usepackage{algorithmic}
\usepackage{array}
\usepackage{cleveref}
\newtheorem{theorem}{Theorem}
\newtheorem{proposition}[theorem]{Proposition}
\newtheorem{definition}[theorem]{Definition}

\newtheorem{lemma}[theorem]{Lemma}

\newtheorem{example}{Example}

\AddToHook{env/proposition/begin}{\crefalias{theorem}{proposition}}
\AddToHook{env/definition/begin}{\crefalias{theorem}{definition}}
\AddToHook{env/corollary/begin}{\crefalias{theorem}{corollary}}
\AddToHook{env/lemma/begin}{\crefalias{theorem}{lemma}}
\AddToHook{env/remark/begin}{\crefalias{theorem}{remark}}
\AddToHook{env/example/begin}{\crefalias{theorem}{example}}
\usepackage{tikz}
\usetikzlibrary{cd}
\usetikzlibrary{shapes.geometric, arrows}
\newcommand{\R}{\mathbb{R}}
\newcommand{\B}{\mathfrak{B}}
\newcommand{\N}{\mathbb{N}}

\newcommand{\bs}[1]{\boldsymbol{#1}}

\newcommand{\sigf}{\sigma_t^{},\sigma_s^{}}
\newcommand{\sig}{\sigma_t^{},\sigma_s}
\newcommand{\sigs}{\sigma_s^{}}
\newcommand{\sigp}{\sigma_t,\sigma_s}

\newcommand{\Br}[1]{\B_{\mathcal{S}_{0,#1}}}

\newcommand{\A}{\mathcal{A}}
\newcommand{\poly}[1]{\A^{#1}}

\newcommand{\sh}{\mathcal{S}_{0,H}}

\title{\LARGE \bf From 2D system to 1D system via restriction to a vertical strip: applications to time-controllability and zero-time-controllability }

\author{Aishwarya Bharti and Debasattam Pal
\thanks{Aishwarya Bharti and Debasattam Pal are with the Department of Electrical Engineering, Indian Institute of Technology Bombay, India. (e-mail: ab@ee.iitb.ac.in, debasattam@ee.iitb.ac.in). Accepted at the 12th Indian Control Conference (ICC-12), 2027.}}%

\begin{document}

\maketitle
\thispagestyle{empty}
\pagestyle{empty}

\begin{abstract}
Motivated by the role of finite-time restriction in one-dimensional (1D) systems, we investigate the restriction of two-dimensional (2D) systems defined over $\N^2$ to a vertical strip. We give a constructive proof that such a restriction induces a 1D system. A Gr\"obner basis-based algorithm is then developed to compute this restricted 1D system from a given 2D system. The framework is then applied to develop an algorithm for testing the time-controllability and zero-time-controllability of scalar autonomous 2D systems; these properties are known to be equivalent to the existence of a deadbeat controller. Finally, we show that, for time-controllable or zero-time-controllable scalar autonomous 2D systems, the induced 1D system generates all the admissible trajectories of the original 2D system, thereby reducing the complexity of trajectory generation.
\end{abstract}

\section{INTRODUCTION}

In this paper, we consider two-dimensional (2D) discrete systems, which are sets of trajectories defined over $\N^2$ and described by the solution set of partial difference equations with constant coefficients. The two independent variables may represent two space coordinates, as in image processing applications or mixed time
and space variables as encountered in seismology, learning and repetitive processes \cite{rocha1990structurethesis, Appli:rogers2015multidimensional}. It is well-known that for one-dimensional (1D) systems, under suitable informativity conditions, the trajectories observed over a finite time duration are sufficient to characterise and analyse the entire system \cite{willems2005note}. This observation forms the basis of subspace identification techniques \cite{van2012subspace} and, more recently, data-driven analysis techniques for 1D systems (see, e.g., \cite{van2020data, van2023behavioral}). Motivated by these developments, it is crucial to investigate whether analogous results hold for 2D systems. Understanding the restriction of trajectories of a 2D system to a finite interval is a crucial step in this direction. In general, however, 2D systems lack an inherent notion of causality in a fixed direction, since both independent variables play equal roles. Nonetheless, in many engineering applications, one of the independent variables represents time, thereby inducing a natural notion of causality \cite{napp2011time, sasane2002time}. Therefore, treating the first independent variable as time, we study the restriction of 2D trajectories to a finite time interval, namely, on the \emph{vertical strip} in $\N^2$ grid given by $\left\{ (h,k)\in \N^2\mid 0\leqslant h\leqslant H,~ k\in \N \right\}$ for $H\in \N$. 

It is shown that the restriction of a 2D system to a vertical strip is a 1D behaviour \cite{rocha1990}, i.e., the trajectories observed in a vertical strip are governed by a set of linear difference equations. However, \cite{rocha1990} only provides the existence of a kernel representation of this restricted 1D behaviour. In this paper, we provide a new constructive proof of this result, thereby devising a method to obtain the representation of the restricted 1D behaviour. Consequently, we employ Gr\"obner basis to present an algorithm that, for a given 2D system, computes the corresponding 1D system governing the restriction of trajectories to a vertical strip.

We use the developed framework to derive an algorithm for testing time-controllability and zero time-controllability of scalar autonomous 2D systems. A 2D system is time-controllable if any two arbitrary trajectories can be patched on an initial vertical strip and on a sufficiently distant half-plane; zero-time-controllability corresponds to such a patching of an arbitrary trajectory with zero trajectory \cite{valcher_bisiacco2013zero}. It is shown in \cite{valcher_bisiacco2013zero} that a 2D system admits a deadbeat controller (DBC) (a controller that acts by full interconnection and makes the resulting system nilpotent) if and only if it is zero-time-controllable. Thus, in essence, we provide a test for the existence of a DBC for scalar autonomous 2D systems. For time-controllable/zero-time-controllable systems, we further show that the associated 1D system obtained through the proposed algorithm is sufficient to generate all admissible trajectories of the original 2D system. This enables data generation by simulating the corresponding 1D system, thereby reducing computational complexity. 

The paper is organised as follows: We begin with preliminaries in \Cref{sec:prelim}. Then, we constructively prove that the restriction of a 2D system to a vertical strip is a 1D system in \Cref{sec:2Dto1D}. In \Cref{sec:GB}, we present a Gr\"obner basis-based algorithm to compute the kernel representation of the restricted 1D behaviour for a given 2D behaviour. Finally, in \Cref{sec:TC-ZTC}, we apply these results to develop an algorithm for testing time-controllability and zero-time-controllability of scalar autonomous 2D systems. 

\paragraph{Notation}\label{notation}
We use $\R$ and $\N$ to denote the sets of real and natural numbers (including $0$), respectively. For two sets $\mathbb{W}$, $\mathbb{T}$, the notation $\mathbb{W}^{\mathbb{T}}$ denotes the set of maps from $\mathbb{T}$ to $\mathbb{W}$, i.e., $\mathbb{W}^\mathbb{T}\coloneqq \left\{f\mid f: \mathbb{T}\to \mathbb{W}\right\}$.
The two-variable polynomial ring in $\sigma_t$ and $\sigma_s$, usually written as $\R[\sig]$, is denoted here by $\mathcal{A}$. Similarly, $\mathcal{A}_{s}$ and $\mathcal{A}_t$ denotes the one-variable polynomial rings $\R[\sigma_s]$ and $\R[\sigma_t]$, respectively. Accordingly, $\R^{m\times n}$, $\A^{m\times n}$ and $\A_s^{m\times n}$ denote the set of all $m\times n$ matrices with entries from $\R$, $\A$ and $\A_s$, respectively. We use $\A^{\bullet\times q}$ when the number of rows is unspecified.
For a matrix $R(\sigf)\in \A^{m\times n}$, expressed as $R(\sigf)\coloneqq \sum_{i=0}^{d}R_i(\sigs)\sigma_t^i$, where $R_i(\sigs)\in \mathcal{A}_s^{m\times n}, ~i=0,\ldots,d$ and $R_{d}(\sigs)\neq 0$, the $\sigma_t$-degree is $d$, denoted as $\deg_{\sigma_t} R(\sigf)=d$. 
For $P_i\in \A^{m_i\times n}$, $i=1,\ldots,k$, $\mathrm{col}(P_1,P_2,\ldots,P_k)$ denotes the $(\sum_{i=1}^{k}m_i)\times n$ matrix obtained by stacking $P_i$'s vertically. 
For $R(\sigf)\in \A^{m\times n}$, its row-span over $\A$ is given by $\mathrm{rowspan}_{\mathcal{A}}R(\sigf)=\left\{x(\sig)R(\sig)\mid x(\sig)\in \mathcal{A}^{1\times m}\right\}$. For a subset $G\in \A^{1\times q}$, we denote by $\langle G\rangle_{\A}$ the $\A$-module generated by elements of $G$. We use $I_{n}$ to denote the $n\times n$ identity matrix.

\section{ Preliminaries}\label{sec:prelim}
In this paper, we consider two-dimensional (2D) systems described by constant-coefficient partial difference equations defined on $\N^2$. Specifically, we consider trajectories $\bs{w}\in (\R^q)^{\N^2}$ satisfying
\begin{align}\label{eq:system}
    \sum_{(i,j)\in \N^2}R_{i,j}\bs{w}(h+i,k+j)=0,~ 
\end{align}
where $R_{i,j}\in\R^{g\times q},~ g\in \N\setminus \{0\}$, and only finitely many $R_{i,j}$ are non-zero.
To express these equations compactly, we define forward shift operators $\sigma_t$ and $\sigma_s$ by: $(\sigma_t \bs{w})(h,k)= \bs{w}(h+1,k)$ and $(\sigma_s \bs{w})(h,k)=\bs{w}(h,k+1)$. Let $R(\sig)\coloneqq\sum\limits_{(i,j)\in \N^2} R_{i,j}\sigma_t^i\sigma_s^j\in \A^{g\times q}$, where $R_{i,j}\in\R^{g\times q},~ g\in \N\setminus \{0\}$ and only finitely many $R_{i,j}$ are non-zero. The corresponding behaviour $\B$ is the set of admissible trajectories satisfying \eqref{eq:system}, given by,  
\vspace{-2mm}
$$
\mathfrak{B}=\!\left\{\bs{w}\in (\R^q)^{\N^2} \mid R(\sig)\bs{w}= 0\right\}=\!\ker R(\sig).
$$
This is called {\em kernel representation} of $\B$ and $R(\sig)$ is called the \emph{kernel representation} matrix.
We next define two algebraic objects associated with a behaviour that we require in the paper.
Note that the kernel representation matrix is not unique to a behaviour. Given a kernel representation matrix $R(\sig)$, the associated \emph{equation module}, denoted by $\mathcal{R}$, is the row-span of $R(\sigp)$ over $\A$. As a consequence of $\R^{\N^2}$ being an injective co-generator module over $\A$, $\mathcal{R}$ is uniquely associated to $\B$ \cite[Theorem 61]{oberst1990multidimensional}. 
This unique correspondence allows us to assign a behaviour $\B$ to a given module $\mathcal{R}$; we use $\B(\mathcal{R})$ to denote this assignment. Thus, $\B(\mathcal{R})\coloneqq \{w\in (\R^q)^{\N^2}\mid r(\sigp)w=0~\forall~ r(\sigp)\in \mathcal{R}\}$.
 Another algebraic object related with $\mathcal{R}$ is the $\A$-module of equivalence classes, where two arbitrary elements $f(\sigp),g(\sigp)\in\poly{1\times q}$ are defined to be equivalent if $f(\sigp)-g(\sigp)\in\mathcal{R}$. This module is called the {\em quotient module} and is denoted by $\mathcal{M}$. The following well-known result, attributed to Malgrange \cite[Page 3]{oberst1990multidimensional}, is a cornerstone in algebraic systems theory and will be used later in this paper. Note that ${\rm Hom}_{\R}(\mathcal{M},\R)$ is the set of all $\R$-linear maps from $\mathcal{M}$ to $\R$, i.e., the dual space of the module $\mathcal{M}$ (viewed as vector space over $\mathbb{R}$).
\begin{proposition}\label{prop:malgrange}
	Let $\mathcal{R}\subseteq\poly{1\times q}$ be a given submodule. Define the quotient module $\mathcal{M}$ as above. Then
	\vspace{-2mm}
	$$
	\B(\mathcal{R})\simeq {\rm Hom}_{\R}(\mathcal{M},\R) \mbox{ as }\A\mbox{-modules}.
	$$
\end{proposition} 

In the sequel, to distinguish the two independent variables, we denote the indexing set by $\N^2=\mathbb{T}\times \mathbb{S}$, where $\mathbb{T}=\N$ and $\mathbb{S}=\N$. Hence, for every $(h,k)\in \N^2$, we have $h\in \mathbb{T}$ and $k\in \mathbb{S}$.
The restriction of a 2D trajectory $\bs{w}\in (\R^q)^{\mathbb{T}\times \mathbb{S}}$ to a subset $\mathcal{S}\subseteq\mathbb{T}\times \mathbb{S}$, denoted by $\bs{w}|_{\mathcal{S}} $, is the map
 $\bs{w}|_{\mathcal{S}}:\mathcal{S}\to\R$ defined by $\bs{w}|_{\mathcal{S}}(h,k)=\bs{w}(h,k)$ for all $(h,k)\in \mathcal{S}$. Accordingly, for a 2D behaviour $\B$ and a subset $\mathcal{S}\subseteq \mathbb{T}\times \mathbb{S}$, the \emph{$\mathcal{S}$-restricted behaviour $\B$}, denoted by $\B_{\mathcal{S}}$, is defined as: 
$
\B_{\mathcal{S}}\coloneqq\left\{\bs{w}|_{\mathcal{S}}:\mathcal{S}\to \R^q \bigm| \bs{w}\in \B\right\}.
$
In this paper, we focus on the restriction of a behaviour to \emph{vertical strips} in $\mathbb{T}\times \mathbb{S}$. Specifically, for $H\in \N$, we consider the strip
\begin{align}\label{eq:sh}
	\sh\subseteq \left\{ (h,k)\in \mathbb{T}\times \mathbb{S}\mid 0\leqslant h\leqslant H , ~k\in \mathbb{S}\right\}. 
\end{align}
 That is, we consider a finite interval in $\mathbb{T}$-direction, as illustrated in \Cref{fig:sh}. 
 
For $\bs{w}\in (\R^q)^{\mathbb{T}\times \mathbb{S}}$, the $\sh$-restricted trajectory $\bs{w}_{\sh}:\sh\to \R^q$ is given by $\bs{w}_{\sh}(h,k)=\bs{w}(h,k)$ for $(h,k)\in \sh$. Thus, the restricted trajectory $\bs{w}_{\sh}$ can be equivalently represented by the column vector ${\rm col}\left(\bs{w}(0,k),\bs{w}(1,k),\ldots,\bs{w}(H,k)\right)\in (\R^{q(H+1)})^\mathbb{S}$. Hence, restricting a trajectory $\bs{w}\in (\R^q)^{\mathbb{T}\times \mathbb{S}}$ to the vertical strip $\sh$ yields a 1D trajectory indexed by $\mathbb{S}$. Consequently, for a given behaviour $\B$, the $\sh$-restricted behaviour $\B_{\sh}$ can be interpreted as a set of 1D trajectories indexed by $\mathbb{S}$. Interestingly, it is shown in \cite[Lemma 1]{rocha1990} that $\Br{H}$ forms a 1D behaviour, i.e., it admits a kernel representation over $\A_s$. However, this result only establishes the existence of a 1D kernel representation matrix over $\A_s$. In the following section, we provide a new constructive proof and obtain the corresponding 1D representation of $\B_{\sh}$.

	\begin{figure}[h]
	\centering
	\begin{tikzpicture}[scale=0.4]
		\foreach \x in {0,1,...,7.5}
		\foreach \y in {0,1,...,7}
		{
			\draw[black!90,opacity=1] (\x,\y) circle (2.5pt);
		}
		\fill [fill=blue!70,opacity=0.2] (0,0) rectangle (3,7.5);
		\draw[->] (0,0)--(7.5,0) node[right]{$\mathbb{T}$};
		\draw[->] (0,0)--(0,7.5) node[left]{$\mathbb{S}$};
		\draw[-] (3,7.5)--(3,0);
		\node at (3,-0.4) {$\scalebox{0.8}{$H$}$};
	\end{tikzpicture}
	\caption{The shaded reason represents $\mathcal{S}_{0,H}$}
	\label{fig:sh}
\end{figure}
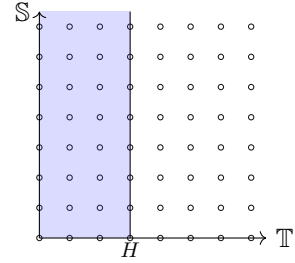


\section{Restriction of $\B$ to the vertical strip $\sh$}\label{sec:2Dto1D}
In this section, for a given 2D system $\mathfrak{B}$, we prove that its restriction to the set $\sh$, defined by \eqref{eq:sh}, namely $\Br{H}$, constitute a 1D behaviour and derive the corresponding representation.
Before presenting the result, we introduce some preliminary concepts. We first define the following subset of $\A$
 $$
 \A_s^{(H)}[\sigma_t]\coloneqq \A_s \oplus \sigma_t\A_s\oplus \ldots\oplus \sigma_t^H \A_s.
 $$
 That is, $\A_s^{(H)}[\sigma_t]$ is composed of all the operators in $\A$ whose $\sigma_t$-degree is at most $H$. Observe that $\A_s^{(H)}[\sigma_t]$ is naturally an $\A_s$-module. Similarly, $ (\A_s^{(H)}[\sigma_t])^{1\times q}$ denotes set of all the operators in $\A^{1\times q}$ whose $\sigma_t$-degree is at most $H$.
 
 Let $\B$ be a 2D system with equation module $\mathcal{R}$. We define $\mathcal{R}_H\coloneqq \mathcal{R}\cap (\A_s^{(H)}[\sigma_t])^{1\times q}$. In other words, $\mathcal{R}_H$ contains all the elements in the equation module $\mathcal{R}$ that have their $\sigma_t$-degree less than or equal to $H$.  
 Let $r(\sigp)\coloneqq 
 \tilde{r}_0(\sigma_s) +\sigma_t \tilde{r}_1(\sigma_s)+\cdots+\sigma_t^H \tilde{r}_H(\sigma_s)$ be an element in $\mathcal{R}_H$. We can associate a 1D operator $r_H(\sigma_s)\in \A_s^{1\times q(H+1)}$ with $r(\sigp)$ using the following map:
 \begin{align}\label{map:psi}
 	\begin{split}
 		\psi: (\A_s^{(H)}[\sigma_t])^{1\times q} &\to \A_s^{1\times q(H+1)}\\
 		r(\sigp)\mapsto& r_H(\sigma_s)\coloneqq \begin{bsmallmatrix}
 			\tilde{r}_0(\sigma_s)&\tilde{r}_1(\sigma_s)&\cdots&\tilde{r}_H(\sigma_s)
 		\end{bsmallmatrix}
 	\end{split}
 \end{align}
 It is obvious that $\psi$ is bijective and satisfies linearity over $\A_s$. The map $\psi$ extends naturally to matrices in $(\A_s^{(H)}[\sigma_t])^{g\times q}$ by acting row-wise.
 We are now ready to show that there exists  $P_H(\sigma_s)\in\A_s^{\bullet\times q(H+1)}$ such that the restriction $\Br{H}$ -- whose elements are identified as 1D trajectories $\bs{w}|_{\sh}(k)\coloneqq{\rm col}\left({w}(0,k),{w}(1,k),\ldots,{w}(H,k)\right)$ -- is equal to the behaviour $\ker P_H(\sigma_s)$, i.e., $\Br{H}=\ker P_H(\sigma_s)$. 
 
\begin{theorem}\label{lem: 2Dto1D}
	Let $\B$ be a 2D behaviour with the equation module $\mathcal{R}$. Let $H\in \N$ and $\sh$ be given by \eqref{eq:sh}. Let $\mathcal{R}_H\coloneqq \mathcal{R}\cap (\A_s^{(H)}[\sigma_t])^{1\times q}$. Then, there exists $P_H(\sigma_s)\in\A_s^{\bullet\times q(H+1)}$ such that the rows of $\psi^{-1}\left(P_H(\sigma_s)\right)$ form the basis of $\mathcal{R}_H$ as an $\A_s$-module. Moreover, $\Br{H}$ is a 1D behaviour such that $\Br{H}=\ker P_H(\sigma_s)$.
\end{theorem}
\begin{proof}
	Via the map $\psi$, defined in \eqref{map:psi}, we identify the $\A_s$-module $(\A_s^{(H)}[\sigma_t])^{1\times q}$ with the free $\A_s$-module $\A_s^{1\times q(H+1)}$. 
	As a result, $(\A_s^{(H)}[\sigma_t])^{1\times q}$ is a Noetherian module over $\A_s$. Thus, the $\A_s$-submodule $\mathcal{R}_H\subseteq (\A_s^{(H)}[\sigma_t])^{1\times q}$ is finitely generated. It then follows that, there exists $P_H(\sigma_s)\in\A_s^{\bullet\times q(H+1)}$ such that the rows of $\psi^{-1}\left(P_H(\sigma_s)\right)$ form the basis for $\mathcal{R}_H$ as an $\A_s$-module. 
	We now show that $\Br{H}$ is a 1D behaviour given by $\Br{H}= \ker P_H(\sigma_s)$.

\textbf{($\boldsymbol{\Br{H}\subseteq \ker P_H(\sigma_s)}$)} Let $\bs{w}\in\B$ be arbitrary; we want to show that the 1D trajectory $\bs{w}|_{\sh}={\rm col}\left(w(0,\bullet),w(1,\bullet),\ldots,w(H,\bullet)\right)$ satisfies $P_H(\sigma_s)\bs{w}|_{\sh}=0$. For this purpose, let us consider an arbitrary row of $P_H(\sigma_s)$,
say $r(\sigma_s)=\begin{bsmallmatrix}
	\tilde{r}_0(\sigma_s)&\tilde{r}_1(\sigma_s)&\cdots&\tilde{r}_H(\sigma_s)
\end{bsmallmatrix}$. Then, we have
\begin{align*}
	\begin{split}
r(\sigma_s)\bs{w}|_{\sh}&=r(\sigs)\operatorname{col}\left(w(0,\bullet),w(1,\bullet),\ldots,w(H,\bullet\right)\\
		&=\left(\psi^{-1}\left(r(\sigma_s)\right)\bs{w}\right)(0,\bullet)\\
		&=0.
	\end{split}
\end{align*}
where the last equation follows from the fact that $\psi^{-1}(r(\sigma_s))\in\mathcal{R}_H\subseteq\mathcal{R}$ and $f(\sigp)\bs{w}\equiv 0$ for all $f(\sigp)\in\mathcal{R}$.

\textbf{($\boldsymbol{\Br{H}\supseteq \ker P_H(\sigma_s)}$)} We take an arbitrary 1D trajectory $\bs{v}\in\left(\R^{q(H+1)}\right)^{\mathbb{S}}$ satisfying $P_H(\sigma_s)\bs{v}=0$ and prove that there exists a $\bs{w}\in\B$ such that $\bs{w}|_{\sh}=\bs{v}$. We use \Cref{prop:malgrange} for this proof. \Cref{prop:malgrange}	applied to the 1D behaviour ${\rm ker~}P_H(\sigma_s)$ brings out the fact that ${\rm ker~}P_H(\sigma_s)$ is isomorphic as an $\A_s$-module to ${\rm Hom}_{\R}\left(\mathcal{Q},\R\right)$, where $\mathcal{Q}$ is the quotient module obtained by factoring $(\A_s^{(H)}[\sigma_t])^{1\times q}$ by $\mathcal{R}_H$, i.e.,
$\mathcal{Q}\coloneqq(\A_s^{(H)}[\sigma_t])^{1\times q}/\mathcal{R}_H$. It can be seen easily that $\mathcal{Q}\hookrightarrow\mathcal{M}$ as $\R$-vector spaces (see, e.g., \cite[Proposition 3.6]{MM:DP:SICON}). Therefore, utilising the exactness of the contravariant functor ${\rm Hom}_{\R}(\bullet, \R)$ we get that ${\rm Hom}_{\R}\left(\mathcal{M},\R\right)\twoheadrightarrow{\rm Hom}_{\R}\left(\mathcal{Q},\R\right)$. Note that, ${\rm Hom}_{\R}\left(\mathcal{M},\R\right)$ is isomorphic to $\B$ according to \Cref{prop:malgrange}. The following commutative diagram thus ensues
(see \cite[Proposition 3.6]{MM:DP:SICON} for more details).

$$
\begin{tikzcd}
	{\rm Hom}_{\R}\left(\mathcal{M},\R\right) \arrow[r, ->>]  & {\rm Hom}_{\R}\left(\mathcal{Q},\R\right)\arrow[d, "\cong"] \\
	\B \arrow[r] \arrow[u, "\cong"'] & \ker P_H(\sigma_s)
\end{tikzcd}
$$
Starting from $\bs{v}\in{\rm ker~}P_H(\sigma_s)$, we can chase the commutative diagram to obtain a $\bs{w}\in\B$ such that $\bs{w}|_{\sh}=\bs{v}$.
\end{proof}	

\section{Constructing a representation of $\Br{H}$ using Gr\"obner basis}\label{sec:GB}
In this section, for a given 2D system $\B$ with the equation module $\mathcal{R}$ and $H\in \N$, we apply \Cref{lem: 2Dto1D} to compute a kernel representation for the restricted 1D behaviour $\B_{\sh}$. By \Cref{lem: 2Dto1D}, we require the basis of the $\A_s$-module $\mathcal{R}_H= \mathcal{R}\cap (\A_s^{(H)}[\sigma_t])^{1\times q}$. However, given a kernel representation of $\B$, whose rows generate $\mathcal{R}$, obtaining a generating set for $\mathcal{R}_H$ is not straightforward. In general, $\mathcal{R}_H$ is not generated by those generators of $\mathcal{R}$ whose $\sigma_t$-degree is at most $H$, since a linear combination of higher-degree generators may cancel the higher powers of $\sigma_t$. Indeed, consider, for example, $\mathcal{R}=\langle \sigma_t\sigma_s^2, \sigma_t^2\sigma_s^2-\sigma_s^2\rangle$ and let $H=0$. Neither generators has $\sigma_t$-degree equal to $0$, nonetheless, $-\sigma_t(\sigma_t\sigma_s^2)+(\sigma_t^2\sigma_s^2-\sigma_s^2)=-\sigma_s^2\in \mathcal{R}_H$. Thus, for $H=0$, $\mathcal{R}_H$ is non-trivial even though none of the given generators belongs to $\mathcal{R}_H$. 

We, therefore, require a special set of generators for $\mathcal{R}$ that reveals the restriction. To obtain such generators, we introduce a specific term ordering, which we refer to as \emph{restriction revealing (RR)} term ordering (see \cite[Chapter 2, Section 2]{cox1994ideals} for the definition of the monomial ordering). With respect to this ordering, the desired generators of $\mathcal{R}$ are the Gr\"obner basis, defined in \Cref{def:GB} below. The RR term ordering is defined as follows: Let $e_i,~i=1,\ldots, q$ be the standard basis vectors in $\R^{1\times q}$. Let $i,j\in \{1,\ldots,q\}$, consider two monomials $\sigma_t^{h_1}\sigma_s^{k_1}e_i$ and $\sigma_t^{h_2}\sigma_s^{k_2}e_j$ in $\A^{1\times q}$. The RR term-ordering is defined as the Term-Over-Position ordering induced by the lexicographic monomial ordering with $\sigma_t\succ \sigma_s$. Specifically, $\sigma_t^{h_1}\sigma_s^{k_1}e_i\succ_{RR} \sigma_t^{h_2}\sigma_s^{k_2}e_j$ if either
\begin{itemize}
	\item $ \sigma_t^{h_1}\sigma_s^{k_1}\succ \sigma_t^{h_2}\sigma_s^{k_2}$ under the lexicographic term-ordering with $\sigma_t\succ \sigma_s$, i.e., $h_1>h_2,$ or $h_1=h_2$ and $k_1>k_2$; or
	\item $ \sigma_t^{h_1}\sigma_s^{k_1}= \sigma_t^{h_2}\sigma_s^{k_2}$ and $i<j$.
\end{itemize} 
We now introduce some terminology required for the definition of the Gr\"obner basis. Let $r(\sig)=\sum\limits_{(i,j)\in \N^2}\alpha_{ij}\sigma_t^i\sigma_s^je_i\in \A^{1\times q}$ where only finitely many $\alpha_{ij}$ are non-zero. The leading monomial of $r(\sig)$, denoted by $LM(r)$, is the largest monomial appearing in $r(\sig)$ with respect to $\succ_{RR}$. The leading coefficient, denoted by $LC(r)$, is the coefficient of $LM(r)$. The leading term, denoted by $LT(r)$, is defined by $LT(r)=LC(r)LM(r)$\cite{cox1994ideals}. 
\begin{definition}\label{def:GB}
	Let $\mathcal{R}\in \A^{1\times q}$ be a submodule. The leading module of $\mathcal{R}$ is defined as $LT(\mathcal{R})=\langle LT(r)\mid r\in \mathcal{R},~ r\neq 0 
\rangle_{\A}$. A finite subset $G\subseteq \mathcal{R}$ is called Gr\"obner basis of $\mathcal{R}$ if $\langle LT(G) \rangle_{\A}= LT(\mathcal{R})$ where $LT(G)=\{LT(g)\mid g\in G\}.$
\end{definition} 

In \Cref{lem:GB_intersection} below, we provide a technique to obtain a set that generates $\mathcal{R}_H$ as an $\A_s$-module using Gr\"obner basis with RR term ordering.
\begin{lemma}\label{lem:GB_intersection}
	Let $\B$ be a 2D behaviour with the equation module $\mathcal{R}\subseteq \A^{1\times q}$ and let $H\in \N$. Define $\mathcal{R}_H=\mathcal{R}\cap (\A_s^{(H)}[\sigma_t])^{1\times q}$. Let $G$ be the Gr\"obner basis of $\mathcal{R}$ with respect to the RR term ordering defined above. Define the set $G_H\coloneqq G\cap (\A_s^{(H)}[\sigma_t])^{1\times q}=\{g_1, g_2,\ldots, g_N\}$ where $N=|G_H|$. Let $d_i=\deg_{\sigma_t}(g_i)$ for $i=1,\ldots, N$. Further, define the set
	$
	\mathcal{G}_H=\left\{\sigma_t^jg_i \mid 0\leqslant j\leqslant H-d_i,~ 1\leqslant i\leqslant N \right\}.
	$
	Then, $\mathcal{R}_H$ equals the $\A_s$-module generated by $\mathcal{G}_H$.
\end{lemma}

\begin{proof}
$\textbf{($\langle\mathcal{G}_H\rangle_{\A_s}\subseteq \mathcal{R}_H$):}$
By definition every $\sigma_t^jg_i\in\mathcal{G}_H$, $i\in \{1,\ldots,N\}$ satisfies $\deg_{\sigma_t} \sigma_t^jg_i\leqslant d_i+j$. Since $ 0\leqslant j\leqslant H-d_i$, it then follows that $\deg_{\sigma_t} \sigma_t^jg_i\leqslant H$, which implies $\mathcal{G}_H\subseteq (\A_s^{(H)}[\sigma_t])^{1\times q}$. Thus, $\mathcal{G}_H\subseteq \mathcal{R}\cap (\A_s^{(H)}[\sigma_t])^{1\times q}= \mathcal{R}_H$. Since, $\mathcal{R}_H$ is an $\A_s$-module, $\langle\mathcal{G}_H\rangle_{\A_s}\subseteq \mathcal{R}_H$.

$\textbf{($\mathcal{R}_H\subseteq \langle\mathcal{G}_H\rangle_{\A_s}$):}$
Let $f\in \mathcal{R}_H$ be arbitrary. It suffices to show that $f$ can be expressed as an $\A_s$-linear combination of elements of $\mathcal{G}_H$.
By multivariate division algorithm \cite[Chapter 5, Theorem 2.5]{cox2006using}, reducing $f$ by the set $G_H=\{g_1,\ldots,g_N\}$ yields
\begin{align}\label{eq:divisionalgo}
	\begin{split}
		f=\sum_{i=1}^{N}\alpha_i g_i+r
	\end{split}
\end{align} 
where $\alpha_i\in \A$, $r\in \A^{1\times q}$, and the following properties hold:
\begin{itemize}
	\item[(i)] $LT(\alpha_i g_i)\preceq_{RR}  LT(f)$ for all $i=1,\ldots, N$.
	\item[(ii)] Either $r=0$ or $r$ is $\R$-linear combination of monomials none of which is divisible by any of $LT(g_i)$, $i=1,\ldots,N$.
\end{itemize} 

We claim that $r=0$. Suppose for the contradiction that $r\neq 0$. Note that we have $r=f-\sum_{i=1}^{N}\alpha_i g_i$. Since $f\in \mathcal{R}$ and $g_i\in \mathcal{R}$ for all $i$, this implies $r\in \mathcal{R}$. Since $G$ is Gr\"obner basis of $\mathcal{R}$, by \Cref{def:GB}, the leading term $LT(r)$ must be divisible by $LT(g)$ for some $g\in G$. 
By the second property of \eqref{eq:divisionalgo}, $LT(r)$ cannot be divisible by $LT(g)$ for any $g\in G_H$. Furthermore, by the first property of the division algorithm, we have $LT(r)\preceq_{RR} \max\{LT(f), LT(\alpha_1g_1),\ldots,LT(\alpha_Ng_N) \}= LT(f)$. Note that in RR term ordering, the $\sigma_t$-degree of any vector $p\in \A^{1\times q}$ is same as $\sigma_t$-degree of $LT(p)$, i.e., $\deg_{\sigma_t} p= \deg_{\sigma_t}(LT(p))$. Since $f\in \mathcal{R}_H$, we have $\deg_{\sigma_t} f=\deg_{\sigma_t}(LT(f))\leqslant H$. Thus, $\deg_{\sigma_t} r=\deg_{\sigma_t}(LT(r))\leqslant H$.  Since for any element $g\in G\setminus G_H$ it must hold that $\deg_{\sigma_t}g>H$, this implies that $LT(r)$ is not divisible by the leading term of any element in $G\setminus G_H$. Thus, $r\in \mathcal{R}$, while $LT(r)$ is not divisible by $LT(g)$ for all $g\in G$, contradicting that $G$ is a Gr\"obner basis of $\mathcal{R}$. Hence, $r=0$.

Now, since $LT(\alpha_i g_i)\preceq_{RR}  LT(f)$ for each $i$, we have $\alpha_i\in \A_s^{(H-d_i)}[\sigma_t]$. That is for $\beta_{ij}\in \A_s$, we can write $\alpha_i=\sum_{j=0}^{H-d_i}\beta_{ij}\sigma_t^j$. Substituting this back into the division equation with $r=0$, we have
$ f=\sum\limits_{\substack{1 \leqslant i \leqslant N \\ 0 \leqslant j \leqslant H-d_i}}
 \beta_{ij} \sigma_t^jg_i,~ \beta_{ij}\in \A_s$. This proves that $f$ is an $\A_s$-linear combination of elements of $\mathcal{G}_H$, completing the proof.
\end{proof}	
\Cref{algo:Bsh} provides algorithm to compute a kernel representation for $\Br{H}$ using \Cref{lem:GB_intersection} and \Cref{lem: 2Dto1D}.
\begin{algorithm}[h]
	\caption{Calculation of kernel representation of $\B_{\sh}$}
	\label{algo:Bsh}
	\textbf{Input:} A 2D behaviour $\B=\ker R(\sig)$ and $H\in \N$\\
	\textbf{Output:} A 1D representation of the restriction $\B_{\sh}$\\
	\begin{enumerate}
		\item Calculate Gr\"obner basis $G$ of the module $\mathcal{R}=\mathrm{rowspan}_{\A}R(\sig)$ with respect to RR term ordering.
		\item Obtain the set $G_H=G\cap \A_s^{(H)}[\sigma_t])^{1\times q}=\{g_1, g_2,\ldots, g_N\}$ where $N=|G_H|$. 
		\item For $d_i=\deg_{\sigma_t}(g_i)$ for $i=1,\ldots, N$, define 
		$
		\mathcal{G}_H=\left\{\sigma_t^jg_i \mid 0\leqslant j\leqslant H-d_i,~ 1\leqslant i\leqslant N \right\}.
		$
		\item Define $P(\sigma_s)= \psi (\mathrm{col}\{\mathcal{G}_H\})$. Then $\B_{\sh}=\ker P(\sigma_s)$.
	\end{enumerate}
\end{algorithm}
\begin{theorem}
	Let $\B$ be a 2D behaviour and let $H\in \N$. Let $\sh$ be given by \eqref{eq:sh}. Then, the output of \Cref{algo:Bsh} is a kernel representation of the restricted 1D behaviour $\B_{\sh}$.
\end{theorem}
\begin{proof}
	By \Cref{lem:GB_intersection}, in step $3)$ of \Cref{algo:Bsh}, we obtain $\mathcal{G}_H$ that generates $\mathcal{R}_H$. Thus, by \Cref{lem: 2Dto1D}, $P(\sigma_s)$, obtained in Step $4)$, satisfies $\B_{\sh}=\ker P(\sigma_s)$.
\end{proof}	
We illustrate \Cref{algo:Bsh} using the following example.
\begin{example}
	Consider a 2D behaviour $\B$ given by
	$\B=\ker \begin{bsmallmatrix}
		\sigma_t^2\sigma_s-\sigma_s&0\\
		\sigma_t^3+3&\sigma_t^2+\sigma_s
	\end{bsmallmatrix}$. Applying \Cref{algo:Bsh}, we calculate the Gr\"obner basis with respect to RR term ordering as $G=\left\{\begin{bsmallmatrix}
	\sigma_t^2\sigma_s-\sigma_s&~0
\end{bsmallmatrix},\begin{bsmallmatrix}
\sigma_t\sigma_s+3\sigma_s&~\sigma_t^2\sigma_s+\sigma_s^2
\end{bsmallmatrix},\begin{bsmallmatrix}
\sigma_t^3+3&~\sigma_t^2+\sigma_s
\end{bsmallmatrix} \right\}$. Let $H=2$, then we obtain that $G_H=\left\{\begin{bsmallmatrix}
\sigma_t^2\sigma_s-\sigma_s&~0
\end{bsmallmatrix},\begin{bsmallmatrix}
\sigma_t\sigma_s+3\sigma_s&~\sigma_t^2\sigma_s+\sigma_s^2
\end{bsmallmatrix} \right\}$. Since all elements of $G_H$ has $\sigma_t$-degree equal to $H$, $\mathcal{G}_H=G_H$. Thus, $P_H(\sigma_s)=\psi (\mathrm{col}\{\mathcal{G}_H\})=\left[\begin{array}{cc|cc|cc}
-\sigma_s&0&0&0&\sigma_s&0\\
3\sigma_s&\sigma_s^2&\sigma_s&0&0&\sigma_s
\end{array}\right]$ and $\Br{H}=\ker P_H(\sigma_s)$.	
\end{example}

\section{Application to time-controllability and zero-time-controllability}\label{sec:TC-ZTC}
In several engineering applications, one of the independent variables represents time, thereby defining a natural notion of causality. This motivates the notion of time-controllability and zero-time-controllability introduced in \cite{valcher_bisiacco2013zero}, which account for the fact that one independent variable is time. It is shown that a 2D behaviour $\B$ admits an admissible deadbeat controller (DBC)\footnote{A DBC for $\B$ is admissible if its sufficient delayed version drives any behaviour trajectory to zero, while preserving its values on a sufficiently large initial strip (see \cite[Section 6]{valcher_bisiacco2013zero} for more details)} if and only if $\B$ is zero-time-controllable \cite[Theorem 6.4]{valcher_bisiacco2013zero}. Consequently, it is important to determine whether a given behaviour is time-controllable or zero-time-controllable. In this section, we use the results developed thus far to derive an algorithm to test these properties in scalar autonomous 2D systems.

Note that a 2D behaviour is autonomous if and only if the kernel representation matrix is a full column rank matrix\cite[Proposition 2.2]{fornasini1993state}. Therefore, in this section, we restrict our attention to behaviours of the form $\B=\ker R(\sig)\subseteq \R^{\N^2}$ where for $g\in \N$, $R(\sig)\in \R^g[\sig]$ is non-zero.

Unlike standard controllability (see \cite[Definition 2]{rocha1991controllability}), which is characterised by the ability to patch two arbitrary trajectories on two distinct sufficiently distant sets of $\N^2$, time-controllability requires patching on two specific subsets of $\N^2$ : an initial vertical strip and a sufficiently distant half plane. More precisely, for  $N,L\in \N$, these sets are 
\begin{align*}
	\begin{split}
		\mathcal{S}_{0,N-1}&=\left\{(h,k)\in \mathbb{T}\times \mathbb{S} \mid 0\leqslant h\leqslant N-1, k\in \mathbb{S} \right\}\\
		\mbox{and }\mathcal{S}_{N+L\to}&=\left\{(h,k)\in \mathbb{T}\times \mathbb{S}\mid h\geqslant N+L, k\in \mathbb{S} \right\}
	\end{split}
\end{align*}
 as illustrated in \Cref{fig:SNSNL}. Zero-time-controllability requires that an arbitrary trajectory can be patched with a zero trajectory on these sets. We provide a formal definitions below \cite[Definition 4.1]{valcher_bisiacco2013zero}.
\begin{definition}
	An 2D behaviour $\B$ is said to be 
	\begin{itemize}
		\item[(i)] 	\emph{time-controllable} if there exists $L\in \N$ such that for every $N\in \N$ and every pair of $\bs{w},\overline{\bs{w}}\in \B$, one can find $\bs{w}^*\in \B$ such that
		\begin{align}
			\bs{w}^*(h,k)&=\bs{w}(h,k) &\forall~ (h,k)\in \mathcal{S}_{0,N-1}\notag\\
			\bs{w}^*(h,k)&=\overline{\bs{w}}(h-N-L,k) &\forall~ (h,k)\in \mathcal{S}_{N+L\to}\notag
		\end{align}
		i.e., $\bs{w}^*|_{\mathcal{S}_{0,N-1}}=\bs{w}|_{\mathcal{S}_{0,N-1}}$ and $\sigma_t^{N+L}\bs{w}^*=\overline{\bs{w}}$.
		\item[(ii)] \emph{zero-time-controllable} if there exists $L\in \N$, such that for every $N\in \N$ and every $\bs{w}\in \B$, one can find $\bs{w}^*\in \B$ such that
		\begin{align}
			\bs{w}^*(h,k)&=\bs{w}(h,k) &\forall~ (h,k)\in \mathcal{S}_{0,N-1}\notag\\
			\bs{w}^*(h,k)&=0 &\forall~ (h,k)\in \mathcal{S}_{N+L\to}\notag
		\end{align}
		i.e., $\bs{w}^*|_{\mathcal{S}_{0,N-1}}=\bs{w}|_{\mathcal{S}_{0,N-1}}$ and $\sigma_t^{N+L}\bs{w}^*=0$.
	\end{itemize}	
\end{definition}

	\begin{figure}[h]
	\centering
	\begin{tikzpicture}[scale=0.4]
		\foreach \x in {0,1,...,6}
		\foreach \y in {0,1,...,7}
		{
			\draw[black!90,opacity=1] (\x,\y) circle (2.5pt);
		}
		\fill [fill=blue!70,opacity=0.2] (0,0) rectangle (2,7.3);
		\draw[->] (0,0)--(6.5,0) node[right]{$\mathbb{T}$};
		\draw[->] (0,0)--(0,7.5) node[left]{$\mathbb{S}$};
		\draw[-] (2,7.5)--(2,0);
		\node at (1.9,-0.4) {$\scalebox{0.8}{$N-1$}$};
		\shade [left color=blue!70, right color=blue!10, opacity=0.2] (4,0) rectangle (6.3,7.3);
		\draw[-] (4,7.5)--(4,0);
		\node at (4.2,-0.4) {$\scalebox{0.8}{$N+L$}$};
        \draw[->]        (5.3,4)   to[out=70,in=170] (6.9,4.5) node[right]{$\mathcal{S}_{N+L\to}$};
         \draw[<-]        (-0.7,4.5)node[left]{$\mathcal{S}_{0,L-1}$}   to[out=10,in=90] (0.7,4) ;
	\end{tikzpicture}
	\caption{Illustration of sets $\mathcal{S}_{0,N-1}$ and $\mathcal{S}_{N+L\to}$}
	\label{fig:SNSNL}
\end{figure}
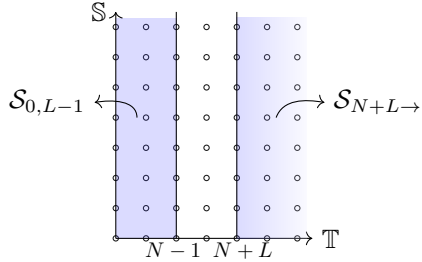

%

It is immediate from the definitions that every time-controllable behaviour is also zero-time-controllable; the converse, however, does not hold. Moreover, every controllable behaviour is time-controllable and hence zero-time-controllable. More importantly, every time-controllable (zero-time-controllable) 2D system can be decomposed into the controllable part and a time-controllable (zero-time-controllable) autonomous part \cite[Theorem 5.1]{valcher_bisiacco2013zero}. Thus, a system with a non-trivial autonomous part may still be time-controllable/zero-time-controllable. This is consistent with the fact that time-controllability is a weaker property than controllability, requiring patching only on special sets. In fact, even for 1D systems, autonomous nilpotent systems are zero-controllable. Thus, their study in scalar autonomous 2D systems provides a crucial starting point for the study in the general system.

The following algebraic characterisation of both time-controllability and zero-time-controllability for scalar autonomous 2D systems follows directly from the corresponding characterisation in \cite[Lemma 4.2]{valcher_bisiacco2013zero} for scalar autonomous systems with kernel representation matrix in $\R[\sig]$. This result forms the basis of the algorithm for testing these properties presented later in this section.

\begin{proposition}\label{prop:TC_scalar}
	Let $\B$ be a 2D scalar autonomous behaviour given by $\B=\ker R(\sig)$ where $R(\sig)\in\R^g[\sig]$. Then,
	\begin{itemize}
		\item[(i)] $\B$ is time-controllable if and only if there exists $\delta(\sigma_s)\in\A_s$ such that $\B=\ker \delta(\sigma_s)$.	
		\item[(ii)] $\B$ is zero-time-controllable if and only if there exists $r\in \N$ and $\delta(\sigma_s)\in \A_s$ such that $\B=\ker \sigma_t^r\delta(\sig)$. 
	\end{itemize}
	
\end{proposition}


The following result is an immediate consequence of \Cref{prop:TC_scalar}. It shows that every admissible trajectory of a time-controllable (zero-time-controllable) system can be obtained from a corresponding 1D system. This observation is particularly useful for trajectory data generation, as the 1D system can be implemented in place of the original 2D system, thereby reducing computational complexity.  

\begin{lemma}\label{lem:dataTC}
	Let $\B$ be a scalar autonomous 2D system. For $t\in \N$, let $\mathcal{S}_t=\left\{(t,k)| k\in \mathbb{S} \right\} $.
	\begin{itemize}
		\item[(i)] $\B$ is time-controllable if and only if there exists $\delta(\sigma_s)\in \A_s$ such that for all $\bs{w}\in \B$, the restriction $\bs{w}|_{\mathcal{S}_t}$ satisfies $\bs{w}|_{\mathcal{S}_t}\in \ker \delta(\sigma_s)$ for all $t\in \N$.
		\item[(ii)]  $\B$ is zero-time-controllable if and only if there exists $r\in \N$ and $\delta(\sigma_s)\in \A_s$ such that for all $\bs{w}\in \B$, the restriction $\bs{w}|_{\mathcal{S}_t}\in \R^{\N}$ for $t< r$ and  $\bs{w}|_{\mathcal{S}_t}\in \ker \delta(\sigma_s)$ for all $t\geqslant r$.
	\end{itemize}
	
\end{lemma}
\begin{proof} 
The claims follow immediately from \Cref{prop:TC_scalar} and the definition of the restriction $\bs{w}|_{\mathcal{S}_t}$.
\end{proof}	

Note that $\A$ is not a principal ideal domain (PID). As a consequence, the equation ideal of 2D system may not be principal, i.e. it may not admit a single generator. Indeed, consider, for example $\B=\ker \begin{bsmallmatrix}
	\sigma_t-1\\
	\sigma_s-2
\end{bsmallmatrix}$, here, there does not exist $r(\sig)\in \R[\sig]$ such that $\B=\ker r(\sig)$. 
However, by \Cref{prop:TC_scalar}, the time-controllable (zero-time-controllable) systems admit the equation ideal with a single generator. Thus, in \Cref{algo:TC-ZTC} for testing time-controllability and zero-time-controllability, we calculate the Gr\"obner basis of the given equation ideal with respect to RR term ordering. By \Cref{algo:Bsh} we then calculate the 1D system corresponding to the first non-zero restriction to vertical lines. We assume that an upper bound on the search $H_{\rm max}$ is known.
To prove the correctness of \Cref{algo:TC-ZTC}, we first establish the following result.

\begin{lemma}\label{lem:foralgo2}
    Let $\B$ be a scalar autonomous 2D behaviour with equation ideal $\mathcal{R}$. Let $H\in \N$ and $\mathcal{R}_H=\mathcal{R}\cap \A_s^{(H)}[\sigma_t]$. Let $\mathcal{G}\subseteq \mathcal{R}_H$ satisfy $\langle\mathcal{G}\rangle_{\A_s}=\mathcal{R}_H$. Then, $\B$ is 
    \begin{itemize}
        \item[(i)] time-controllable if and only if, for $H=0$, $\mathcal{R}=\langle\mathcal{G}\rangle_{\A}$. 
        \item[(ii)] zero-time-controllable if and only if there exists $r\in \N$ such that $\mathcal{R}_H=\varnothing$ for every $H<r$, while, for $H=r$, $\mathcal{R}=\langle\mathcal{G}\rangle_{\A}$ and $\sigma_t^{-r}g\in \A_s$ for all $g\in \mathcal{G}$. 
    \end{itemize}
\end{lemma}
\begin{proof}
The `only if' part in both cases follows immediately by \Cref{prop:TC_scalar}; we therefore prove the `if' direction.
   $(i)$ Suppose the stated condition holds. Since $\mathcal{R}_H$ is an $\A_s$-module and $\A_s$ is PID, there exists $\delta(\sigma_s)\in \A_s$ such that $\mathcal{R}=\langle\delta(\sigma_s)\rangle_{\A}$ and therefore $\B=\ker \delta(\sigma_s)$.

    $(ii)$ Suppose the stated condition holds. Then, every $g\in \mathcal{G}$ is of the form $g=\sigma_t^r f(\sigma_s)$ for some $f(\sigma_s)\in \A_s$. Since $\A_s$ is a PID, there exists $\delta(\sigma_s)\in \A_s$ such that $\mathcal{R}=\langle\sigma_t^r\delta(\sigma_s)\rangle_{\A}$. Therefore $\B=\ker \sigma_t^r\delta(\sigma_s)$.
\end{proof}

\begin{algorithm}
	\caption{Test for time-controllability and zero-time-controllability}
	\label{algo:TC-ZTC}
	\begin{algorithmic}
		
		\STATE \textbf{Input:} A scalar autonomous 2D behaviour $\B=\ker R(\sig)$ and $H_{\rm max}\in \N$
		\STATE \textbf{Output:} Whether $\B$ is time-controllable (TC) / zero-time-controllable(ZTC) or not; if yes, then return $\delta(\sigma_s)\in \A_s$ satisfying \Cref{lem:dataTC}.
		
		\STATE \textbf{Solution:} Calculate Gr\"obner basis $G$ of the module $\mathcal{R}=\mathrm{rowspan}_\A R(\sig)$ w.r.t. RR term ordering.
		\STATE $temp \gets \text{0}$
		\FOR{$H = 0$ to $H_{\max}$}
		\STATE Compute $G_H=G\cap \A_s^{(H)}[\sigma_t]$ until $G_H \neq 0$
		\IF{$\langle G_H\rangle_\A =\mathcal{R}$}
		\IF{$H = 0$}
		\STATE \text{$\B$ is TC and $\delta(\sigma_s)=G_H$;} $temp \gets \text{1}$
		\ELSE 
		\IF{$\sigma_t^{-H} G_H\in \A_s$}
		\STATE \text{$\B$ is ZTC and $\delta(\sig)=\sigma_t^{-H}G_H$;} $temp \gets \text{1}$
		\ENDIF
		\ENDIF
		\ENDIF
		\ENDFOR
		\IF{$temp = \text{0}$}
		\STATE \text{$\B$ is neither TC nor ZTC.}
		\ENDIF
	\end{algorithmic}
\end{algorithm}

\begin{theorem}\label{thm:algo2}
		Let $\B$ be a scalar autonomous system. Then \Cref{algo:TC-ZTC} provides a test for time-controllability and zero-time-controllability. Moreover, if $\B$ is TC or ZTC, the algorithm returns $\delta(\sigma_s)$ satisfying \Cref{lem:dataTC}.
\end{theorem}
\begin{proof}
The correctness of the Algorithm for the testing for TC and ZTC follows from \Cref{lem:foralgo2}. Moreover, since $G$ is the Gr\"obner basis of the equation ideal $\mathcal{R}$, \Cref{lem:foralgo2} and \Cref{lem:GB_intersection} imply that when the system is TC or ZTC, $G_H$ in the \Cref{algo:TC-ZTC} is the Gr\"obner basis of $\mathcal{R}_H$. Hence, the corresponding $\delta(\sigma_s)$ in \Cref{algo:TC-ZTC} satisfies \Cref{lem:dataTC}.
\end{proof}	

%
We illustrate \Cref{algo:TC-ZTC} using the following examples.
\begin{example}
	Consider a scalar autonomous 2D system 
	$\B=\ker \begin{bsmallmatrix}
		\sigma_t\sigma_s^2\\
		\sigma_s^2-\sigma_t^2\sigma_s^2
	\end{bsmallmatrix}$.
Here the equation ideal is $\mathcal{R}=\langle\sigma_t\sigma_s^2, \sigma_s^2-\sigma_t^2\sigma_s^2 
\rangle_{\A}$. Applying \Cref{algo:TC-ZTC}, the Gr\"obner basis of $\mathcal{R}$ is $G=\sigma_s^2$. Hence, for $H=0$, we obtain $G_H=\sigma_s^2\neq 0$ and also $\langle G_H\rangle_{\A}=\langle G\rangle_{\A}=\mathcal{R}$. Therefore, by \Cref{algo:TC-ZTC}, $\B$ is time-controllable. Furthermore, every $\bs{w}\in \B$, we have $\bs{w}_{\mathcal{S}_t}\in \ker \sigma_s^2$ for all $t\in \N$.
\end{example}
\begin{example}
	Consider a scalar autonomous 2D system
	$\B=\ker \begin{bsmallmatrix}
		\sigma_t^2\sigma_s\\
		\sigma_t\sigma_s^2+\sigma_t\sigma_s\\
		\sigma_t\sigma_s^2
	\end{bsmallmatrix}$.
	Here, the equation ideal is $\mathcal{R}=\langle	\sigma_t^2\sigma_s,
	\sigma_t\sigma_s^2+\sigma_t\sigma_s,
	\sigma_t\sigma_s^2 
\rangle_{\A}$. Applying \Cref{algo:TC-ZTC}, the Gr\"obner basis of $\mathcal{R}$ is $G=\sigma_t\sigma_s$. Here, with $H=0$, we have $G_H= 0$, but for $H=1$, we obtain $G_H=\sigma_t\sigma_s\neq 0$. Moreover, for $H=1$, $\langle G_H\rangle_{\A}=\langle G\rangle_{\A}=\mathcal{R}$. Next, since $\sigma_t^{-1}G_H=\sigma_s\in \A_s$, by \Cref{algo:TC-ZTC} $\B$ is zero-time-controllable. Consequently, by \Cref{lem:dataTC}, for all $\bs{w}\in \B$, $\bs{w}_{\mathcal{S}_t}\in \ker \sigma_s^2$ for all $t\geqslant 1$ and $\bs{w}(0,\cdot)$ is free.
\end{example}

\section{Conclusion}\label{sec:conclusion}
In this paper, we presented a constructive proof that restricting a 2D system to a vertical strip yields 1D behaviour. Based on this result, we introduced the restriction-revealing (RR) term ordering and developed a Gr\"obner basis-based algorithm for computing the corresponding restricted 1D system for a given 2D system. For scalar autonomous 2D systems, we further applied this framework to devise an algorithm for testing time-controllability and zero-time-controllability. When these properties hold, the algorithm also computes an associated 1D system that determin entire trajectory of the original 2D system. Future work will extend the proposed approach to a general 2D system.

\bibliographystyle{IEEEtran}
\bibliography{ref}

\end{document}